\documentclass[11pt]{amsart}

\usepackage{amsmath}
\usepackage{amsthm}
\usepackage{amssymb}
\usepackage{fullpage}
\usepackage{mathtools}
\usepackage{color}
\usepackage{comment}

\newtheorem{Thm}{Theorem}[section]
\newtheorem{thm}[Thm]{Theorem}
\newtheorem{prop}[Thm]{Proposition}
\newtheorem{coro}[Thm]{Corollary}

\newtheorem{lemma}[Thm]{Lemma}

\theoremstyle{definition}
\newtheorem{defi}[Thm]{Definition}

\newtheorem{exam}[Thm]{Example}

\newcommand{\F}{\mathbb{F}}

\newcommand{\N}{\mathbb{N}}

\newcommand{\supp}{\operatorname{supp}}
\newcommand{\Cs}{\mathcal{C}}
\newcommand{\Hs}{\mathcal{H}}
\newcommand{\RM}{\mathcal{RM}}
\newcommand{\Ps}{\mathcal{P}}

\newcommand{\mat}[1]{\left(\begin{matrix}#1\end{matrix} \right)}

\title{Batch Codes from Hamming and Reed-M{\"u}ller Codes}
\author{Travis Baumbaugh}
\address{Department of Mathematical Sciences, Clemson University,
  Clemson, SC 29634}
\email{tbaumba@g.clemson.edu}

\author{Yariana Diaz} 
\address{Department of Mathematics \& Statistics, Amherst College,
  Amherst, MA 01002
}
\email{ydiaz18@amherst.edu}

\author{Sophia Friesenhahn}
\address{Department of Mathematics, Willamette University, Salem, OR 97301}
\email{sophmf@gmail.com}

\author{Felice Manganiello}
\address{Department of Mathematical Sciences, Clemson University,
  Clemson, SC 29634}
\email{manganm@clemson.edu}

\author{Alexander Vetter}
\address{Department of Mathematics \& Statistics, Villanova
  University, Villanova, PA 19085}
\email{avetter@villanova.edu}

\begin{document}
\maketitle

\pagestyle{plain}

\begin{abstract}
\noindent Batch codes, introduced by Ishai \textit{et al.} encode a
string $x \in \Sigma^{k}$ into an $m$-tuple of strings, called
buckets. In this paper we consider multiset batch codes wherein a set
of $t$-users wish to access one bit of information each from the
original string. 
We introduce a concept of optimal batch codes.  We first show that
binary Hamming codes are optimal batch codes. The main body of this
work provides batch properties of Reed-M\"uller codes. We look at
locality and availability properties of first order Reed-M{\"u}ller
codes over any finite field. We then show that binary first order
Reed-M{\"u}ller codes are optimal batch codes when the number of users
is 4 and generalize our study to the family of binary Reed-M{\"u}ller
codes which have order less than half their length.

\end{abstract}

\section{Introduction}\label{s1}

\indent Consider the situation where a certain amount of data, such as information to be downloaded, is distributed over a number of devices. We could have multiple users who wish to download this data. In order to reduce wait time, we look at locally repairable codes with availability as noted in \cite{dimakis}. A locally repairable code, with locality $r$ and availability $\delta$, provides us the opportunity to reconstruct a particular bit of data using $\delta$ disjoint sets of size at most $r$ \cite{skachek}. When we want to reconstruct a certain bit of information, this is the same as a Private Information Retrieval (PIR) code. However, we wish to examine a scenario where we reconstruct not necessarily distinct bits of information.

\indent A possible answer to the more complex scenario seems to be a family of codes called batch codes, introduced by Ishai \textit{et al.} in \cite{ucla}. Batch codes were originally studied as a schematic for distributing data across multiple devices and minimizing the load on each device and total amount of storage. We study $(n,k,t,m,\tau)$ batch codes, where $n$ is the code length, $k$ is the dimension of the code, $t$ is the number of bits we wish to retrieve, $m$ is the number of buckets, and $\tau$ is the maximum number of bits used from each bucket for any reconstruction of $t$ bits. In this paper we seek to minimize the number of devices in the system and the load on each device while maximizing the amount of reconstructed data. In other words, we want to minimize $m\tau$ while maximizing $t$.

 In Section \ref{s2}, we formally introduce batch codes and
summarize results from previous work on batch codes. We then introduce
the concepts of locality and availability of a code. We conclude the
section by introducing a concept of optimal batch codes.

 After the background, we study the batch properties of binary
Hamming codes and Reed-M{\"u}ller codes. Section \ref{s3} focuses on
batch properties of binary Hamming codes. We show that Hamming codes
are optimal $(2^{s-1}, 2^{s}-1-s, 2, m, \tau)$ batch codes for
$m,\tau\in\N$ such that $m\tau = 2^{s-1}$.

 Section \ref{s4} is the main body of this work and provides
batch properties of Reed-M\"uller codes. We first study the induced batch properties
of a code $\Cs$ given that $\Cs^{\perp}$ is of a $(u \mid u+v)$-code
construction with determined batch properties. 
In Section \ref{s43} we study the locality and availability properties
of first order Reed-M{\"u}ller codes over any finite field. We find
that the locality of $\RM_q(1,\mu)$ is $2$ when $q\ne2$ and
$3$ when the $q=2$. Furthermore, we also show that its availability  is
$\left\lfloor\frac{q^{\mu}-1}{2}\right\rfloor$ when $q\ne 2$, whereas
when $q=2$, the availability is
$\frac{2^\mu-1}3$ if $\mu$ is even and at least $\frac{2^\mu-4}4$ otherwise.
 In Section
\ref{s432} we show that binary first order Reed-M{\"u}ller codes are
optimal batch codes for $t=4$.  We first
look at the specific $\RM(1,4)$ case and achieve parameters
$(16,5,4,m,\tau)$ such that $m\tau = 10$. We then prove a general
result that any Reed-M{\"u}ller code with $\rho=1$ and $\mu\geq 4$ has
batch properties $(2^{\mu},\mu+1,4,m,\tau)$ for any
$m,\tau\in\mathbb{N}$ such that $m\tau=10$.

\indent Finally, in Section \ref{s44} we generalize our study of Reed-M{\"u}ller codes and look at properties of $\RM(\rho,\mu)$ for all values of $\rho$ and conclude our study by presenting batch properties $(2^{\mu},k,4,m,\tau)$ such that $m\tau = 10\cdot2^{2\rho -2}$ for $\RM(\rho',\mu)$ where $\mu\in\{2\rho+2,2\rho+3\}$ and $\rho'\le\rho$.

\section{Background}\label{s2}
\indent In 2004, Ishai \textit{et al.} \cite{ucla} introduced the following definition of batch codes:
\begin{defi}\label{batch1}
An $(n,k,t,m,\tau)$ batch code over an alphabet $\Sigma$ encodes a string $x\in\Sigma^{k}$ into an $m$-tuple of strings, called buckets, of total length $n$ such that for each $t$-tuple of distinct indices, $i_{1},\dots,i_{t}\in[k]$, the entries $x_{i_{1}},\dots,x_{i_{t}}$ can be decoded by reading at most $\tau$ symbols from each bucket. 
\end{defi}
\indent We can view the buckets as servers and the symbols used from each bucket as the maximal load on each server. In the above scenario, a single user is trying to reconstruct $t$ bits of information. This definition naturally leads to the concept of multiset batch codes which have nearly the same definition as above, but the indices $i_{1},\dots,i_{k}\in[k]$ are not necessarily distinct. This means we have $t$ users who each wish to reconstruct a single element. This definition in turn relates to private information retrieval (PIR) codes which are similar to batch codes but instead look to reconstruct the same bit of information $t$ times. Another notable type of batch code defined in \cite{ucla} is a primitive multiset batch code where the number of buckets is $m=n$.

The following are useful lemmas proven in \cite{ucla}:
\begin{lemma}\label{lemucla}
An $(n,k,t,m,\tau)$ batch code for any $\tau$ implies an $(n\tau,k,t,m\tau,1)$ batch code.
\end{lemma}
\begin{lemma}\label{lemucla2}
An $(n,k,t,m,1)$ batch code implies an $(n,k,t,\lceil\frac{m}{\tau}\rceil,\tau)$ batch code.
\end{lemma} 
\indent Much of the related research involves primitive multiset batch codes with a systematic generator matrix. In \cite{ucla}, the authors give results for some multiset batch codes using subcube codes and Reed-M{\"u}ller codes. They use a systematic generator matrix, which often allows for better parameters. Their goal was to maximize the efficiency of the code for a fixed number of queries $t$. The focus of research on batch codes then shifted to combinatorial batch codes. These were first introduced by \cite{paterson}. They are replication based codes using various combinatorial objects that allow for efficient decoding procedures. We do not consider combinatorial batch codes but some relevant results can be found in \cite{paterson}, \cite{ruj}, \cite{tuza}, and \cite{gal}. \\
\indent Next, the focus of research turned to linear batch codes, which use classical error-correcting codes. The following useful results are proven in \cite{lipmaa}:
\begin{thm}\label{thma}
Let $\Cs$ be an $[n,k,t,n,1]$ linear batch code over $\mathbb{F}_{2}$ with generator $G$. Then, $G$ is a generator matrix of the classical error-correcting $[n,k,d]_{2}$ linear code where $d\ge t$. 
\end{thm}
\begin{thm}\label{thmb}
Let $\Cs_{1}$ be an $[n_{1},k,t_{1},n,1]_{q}$ linear batch code and $\Cs_{2}$ be an $[n_{2},k,t_{2},n_{2},1]_{q}$ linear batch code. Then, there exists an $[n_{1}+n_{2},k,t_{1}+t_{2},n_{1}+n_{2},1]_{q}$ linear batch code.
\end{thm}
\begin{thm}\label{thmc}
Let $\Cs_{1}$ be an $[n_{1},k_{1},t_{1},n_{1},1]_{q}$ linear batch code and $\Cs_{2}$ be an $[n_{2},k_{2},t_{2},n_{2},1]_{q}$ linear batch code. Then, there exists an $[n_{1}+n_{2},k_{1}+k_{2},min(t_{1},t_{2}),n_{1}+n_{2},1]_{q}$ linear batch code.
\end{thm}
\indent Because of the vast amount of information on classical error-correcting codes, we use these as our central focus in this paper. As is often the case, studying the properties of the dual codes helps us find efficient batch codes.

Next, \cite{thomas} considers restricting the size of reconstruction sets. These are similar to codes with locality and availability:
\begin{defi}\label{loc}
A code $\Cs$ has locality $r\geq1$ if for any $c\in\Cs$, any entry in $c$ can be recovered by using at most $r$ other entries of $c$.
\end{defi}
\begin{defi}\label{avail}
A code $\Cs$ with locality $r\geq1$ has availability $\delta\geq1$ if for any $c\in\Cs$, any entry of $c$ can be recovered by using one of $\delta$ disjoint sets of at most $r$ other entries of $y$
\end{defi}
\indent In \cite{thomas}, linear multiset batch codes with restricted reconstruction sets were presented. The restriction on the size of reconstruction sets can be viewed as trying to minimize total data distribution. We restrict the size of our reconstruction sets to the locality of the code. By using this restriction, we find multiset batch codes with optimal data distribution. For cyclic codes, the locality can be derived from the following in \cite{hu16}:

\begin{lemma}\label{haminfo}
Let $\Cs$ be an $[n, k, d]$ cyclic code, and let $d'$ be the minimum distance of its dual code $\Cs^{\perp}$. Then, the code $\Cs$ has all symbol locality $d'-1$.
\end{lemma}
This relies on each entry being in the support of a minimal weight
dual code word. We generalize this lemma to the following one.

\begin{lemma}\label{RMinfo}
  Let $\Cs\subseteq \F_q^n$ be a linear code and let $d^\prime$ be the minimum
  distance of $\Cs^\perp$. If $\Cs^\perp$ is generated by its minimum
  weight codewords and
  \begin{equation}\label{indep}
    \bigcup_{\lambda\in \Cs^\perp}\supp(\lambda)=[n],
\end{equation}
then  $\Cs$ has all symbol locality $d'-1$. 
\end{lemma}

\begin{proof}
  Condition \eqref{indep} implies that no coordinate of $\Cs$ is
  independent on the others. If the minimum weight codewords generate
  $\Cs^\perp$, then each coordinate of $\Cs$ is
  in the support of at least one minimum weight codeword of $\Cs^\perp$.  This
  implies the all symbol locality $d'-1$ of $\Cs$.
\end{proof}

Condition \eqref{indep} is a reasonable condition for a code. Without
it the code $\Cs$ would have non recoverable coordinates.
   
We give here a bound that relates the locality property of a
linear code with its batch properties.

\begin{lemma}\label{bound}
Let $\Cs$ be an $[n,k,t,m,\tau]$ linear batch code with minimal locality $r$. Then it holds that
\begin{equation}\label{mtaubound}
m\tau\geq(t-1)r +1.
\end{equation}
\end{lemma}
\begin{proof}

We consider such a code $\Cs$. If each entry has at least one reconstruction set with fewer than $r$ elements, then by the definition of locality, $\Cs$ has locality less than $r$, a contradiction to $r$ being the minimal locality. Therefore, there exists at least one entry for which all recovery sets are of size at least $r$. If we wish to recover this entry $t$ times, then we may read the entry itself and then make use of $t-1$ disjoint recovery sets, each of size $r$. This implies reading at least $(t-1)r+1$ entries, and since we may read only $\tau$ entries from each of the $m$ buckets, we must have that $m\tau\ge(t-1)r+1$.
\end{proof}
From the perspective of individual devices storing bits of data,
$m\tau$ represents the total amount of data read to provide $t$ pieces
of the original data. To minimize bandwidth usage in the case where
the entries of a codeword represent nodes on a network, we must
minimize $m\tau$.

\begin{defi}
  A $[n,k,t,m,\tau]$ linear batch code $\Cs$ with minimal locality $r$
  is optimal if it satisfies Condition \eqref{mtaubound} with equality. 
\end{defi}

We show now that Hamming codes are optimal linear batch codes.

\section{Hamming Codes}\label{s3}

Hamming codes were first introduced in 1950 by Richard Hamming. In what follows, we consider binary Hamming odes over $\mathbb{F}_{2}$. The parameters of binary Hamming codes are shown in \cite{bluebk}.
\begin{defi}\label{ham}
For some $s\geq 2$, let $H \in \F_{2}^{2^{s}-1 \times s}$ be a matrix whose columns are all of the nonzero vectors of $\mathbb{F}_{2}^{s}$. Let $n=2^{s}-1$. We use $H$ as our parity check matrix and define the binary Hamming Code:
\[\Hs_{s} \coloneqq \{c\in\mathbb{F}_{2}^{n}\mid cH^{T}=0\}\]
\end{defi}
\indent It is well-known that $\Hs_{s}$ is a $[2^{s}-1,2^{s}-1-s,3]$ cyclic code. Its dual code, the simplex code, is a $[2^{s}-1,s,2^{s-1}]$ cyclic code. Thus, by Lemmma \ref{haminfo}, the locality of $\Hs_{s}$ is $2^{s-1}-1$. We now present the batch properties of binary Hamming Codes.


Hamming codes are optimal linear batch codes.
\begin{thm}\label{hambatch}
A binary $[n=2^{s}-1,k=2^{s}-1-s]$ Hamming code is an optimal batch code with properties $[2^{s}-1,2^{s}-1-s,2,m,\tau]$, for any $m,\tau\in\mathbb{N}$ such that $m\tau=2^{s-1}$.
\end{thm}

\begin{proof}
First, note that $m\tau\geq (2-1)(2^{s-1}-1) + 1= 2^{s-1}$. The buckets for $m = 2^{s-1}$, $\tau = 1$ are constructed as follows. Let $H$ be the parity check matrix of a binary Hamming code, $\Hs$, with columns $h_{j} \in \F_{2}^{s}$ for $1\leq j\leq n$. If $h_{a}+h_{b}= 1$ (the all ones column), then we place $a$ and $b$ into the same bucket. Note that because $h_{\ell} = 1$ in $H$, $\ell$ is placed into its own bucket. Let $r_{d}\in\F_{2}^{n}$ be the rows in $H$ such that $1\leq d\leq (n-k)=s$. For any $c\in\Hs$, $r_{d} \cdot c^{T}=0$, and thus $\sum_{i\in\supp(r_{d})}c_{i}= 0$. As a result of this construction, if $a,b\in\supp(r_{d})$, then entry $d$ of $h_a+h_b$ is $0$, so $a$ and $b$ cannot be in the same bucket. Therefore, $\sum_{i\in \supp(r_{d})} c_{i}= 0$ only involves bits in separate buckets. Hence, any bit in a codeword can be written as a linear combination of bits in separate buckets. Now, we show how to reconstruct any two bits of information.
\begin{itemize}
\item Case 1: If $a$ and $b$ are in separate buckets, then use $c_a$ and $c_b$.

\item Case 2: Suppose $a$ and $b$ are in the same bucket. We can take $c_a$ itself. To construct $c_b$, we choose an $r_d$ such that $b\in\supp(r_d)$. Then, we can write $c_b=\sum_{i\in\supp(r_d)\setminus\{b\}}c_{i}$, which we know only contains bits in disjoint buckets.
\end{itemize}

Every bucket has cardinality $2$ aside from the bucket corresponding to the all ones column in $H$, so this construction gives us exactly $2^{s-1}$ buckets. Thus, we have shown that the batch properties hold for $m=2^{s-1}$ and $\tau = 1$. Further, Lemma \ref{lemucla2} implies that this is true for any $m,\tau$ such that $m\tau=2^{s-1}$.

Note that the locality of $\Hs$ is $2^{s-1}-1$, and therefore, $t=2$ is also maximal. Suppose instead that we could have $t\ge3$. Then in particular, each entry must be reconstructible at least $3$ times. We may take the entry itself, but then there must be at least $2$ other reconstruction sets used which are disjoint and of size $2^{s-1}-1$. These would correspond to two codewords in the dual code of weight $2^{s-1}$ with the intersection of their support being only the given entry. The sum of these codewords will thus have weight $2^{s-1}+2^{s-1}-2=2^{s}-2$. However, the all ones vector is also in the dual code. Adding this vector to the sum will produce a codeword of weight one, a contradiction. Thus, $t=2$ is maximal.
\end{proof}

\begin{exam}\label{hamex}
We now give an example for $s=3$. This Hamming Code is a $[7,4]$-linear code, and the dual code is a $[7,3]$-linear code.
\end{exam}
The parity check matrix $H$ is as follows:
\[H=
\mat{
1 & 0 & 1 & 0 & 1 & 0 & 1\\
0 & 1 & 1 & 0 & 0 & 1 & 1\\
0 & 0 & 0 & 1 & 1 & 1 & 1
}
\]
Thus, the buckets are:
\begin{center}
$\{1,6\}, \{2,5\},\{3,4\},\{7\}$
\end{center}

We note that for general batch codes we are only interested in reconstructing information bits. However, we are able to obtain any pair of bits in the codeword. Additionally, we note that although $m \tau$ is optimized, we wish to find batch codes where $t>2$. We desire a larger value as we are concerned with practical applications, and the goal is to quickly distribute data. Thus we move on to Reed-M{\"u}ller codes, where we are able to obtain larger $t$ values.

\section{Reed-M{\"u}ller Codes}\label{s4}

Reed-M\"uller codes are well known linear codes. We give some basic
properties of these codes but an interested reader can find more
information in \cite{as92}.

\begin{defi}\label{RMdef}
Let $\mathbb{F}_{q}[X_{1},\dots,X_{\mu}]$ be the ring of polynomials in $\mu$ variables with coefficients in $\mathbb{F}_{q}$ and let $\F_q^{\mu}=\{P_{1},\dots,P_{n}\}$ (so $n=q^{\mu}$). The $q$-ary Reed-M{\"u}ller code, $\RM_{q}(\rho,\mu)$ is defined as:
\[\RM_{q}(\rho,\mu) \coloneqq \{(f(P_{1}),\dots,f(P_{n}))\mid f\in\F_{q}[X_{1},\dots,X_{\mu}]^{\rho}\},
\]
where $\F_{q}[X_{1},\dots,X_{\mu}]^{\rho}$ is the set of all multivariate polynomials over $\F_q$ of total degree at most $\rho$.
\end{defi}

It is known that if $\rho<\mu(q-1)$ then the dual of a Reed-M\"uller
code $\RM_q(\rho,\mu)$ is the Reed-M\"uller code
$\RM_q(\mu(q-1)-1-\rho,\mu)$ \cite{as92}.

In the binary case, Reed-M\"uller codes can be equivalently defined
using the $(u \mid u + v)$-code construction. For completeness, we first give a description of the  $(u\mid u+v)$-code construction and the related generator matrix, which can be found in \cite{bluebk}.
\begin{defi}\label{uv}
 Given two linear codes $\Cs_{1}, \Cs_{2}$ with identical alphabets and block lengths, we construct a new code $\Cs$ as follows:
\[\Cs \coloneqq \{(u \mid u + v) \mid u \in \Cs_{1}, \, v \in \Cs_{2} \}.\]
\end{defi}

Let $G$, $G_{1}$, and $G_{2}$ be the generator matrices for the codes $\Cs$, $\Cs_{1}$, and $\Cs_{2}$, respectively, where $\Cs$ is obtained from $\Cs_{1}$ and $\Cs_{2}$ via the $(u \mid u+v)$-construction. Then we have
\[G \coloneqq \mat{
    G_{1} & G_{1} \\
    0 & G_{2}
}\]
\medskip
From this, we have the following proposition.
\begin{prop}\label{uvprop}
Let $\Cs_{1}$ be an $[n, k_{1}, d_{1}]$-linear code, $\Cs_{2}$ an $[n, k_{2}, d_{2}]$-linear code, and $\Cs$ the code obtained from $\Cs_{1}, \Cs_{2}$ via the $(u \mid u+v)$-construction. Then, $\Cs$ is an $[2n,k,d]$-code where $k=k_{1}+k_{2}$ and $d=\min\{2d_{1},d_{2}\}$.
\end{prop}

Later, our focus will be on $q=2$, so when referring to
$\RM_{2}(\rho,\mu)$ we omit the $2$ for convenience. We obtain the
following equivalent definition of a binary Reed-M\"uller code.

\begin{defi}\label{bRMdef}
Let $\rho<\mu$. A binary Reed-M{\"u}ller code $\RM(\rho,\mu)$ is defined as follows:
\[\RM(\rho,\mu) \coloneqq \{ (u \mid u + v) \mid u \in \RM(\rho, \mu-1), \, v \in \RM(\rho-1, \mu-1) \}\]
where $\RM(0,\mu) \coloneqq 1$ of length $2^{\mu}$, and $\RM(\mu,\mu)
\coloneqq I_{2^{\mu}}$. 
\end{defi}

As a consequence if $G_{\rho,\mu}$ is the generator matrix of the code
$\RM(\rho,\mu)$, then 
\[G_{\rho,\mu} \coloneqq 
\mat{
    G_{\rho,\mu-1} & G_{\rho,\mu-1} \\
    0 & G_{\rho-1,\mu-1}
}\]



We now examine the batch properties of the {$(u\mid v+u)$}-code construction. The first notable result comes from codes that are contained in other codes:
\begin{thm}\label{subspacethm}
Let $\Cs_{1},\Cs_{2}$ be codes of length $n$ and dimension $k_1$ and $k_2$, respectively such that $\Cs_2\subseteq\Cs_1$. If $\Cs_{1}$ is a $(n,k_{1},t,m,\tau)$ batch code, then $\Cs_{2}$ is a $(n,k_{2},t,m,\tau)$ batch code.
\end{thm}
\begin{proof}
Note that $\Cs_{2}^{\perp}\subseteq\Cs_{1}^{\perp}$ because $\Cs_{1}\subseteq\Cs_{2}$. Therefore the same parity check equations for $\Cs_{2}$ apply to $\Cs_{1}$. Thus, to recover information in $\Cs_1$, we may use the same parity check equations we would in $\Cs_2$, which implies $\Cs_{1}$ is at least a $(n,k_{1},t,m,\tau)$ batch code.
\end{proof}

We now introduce results for a $(u\mid u+v)$-code construction.
\begin{thm}\label{uvthm}
Let $n,k_1,k_2\in\N$ such that $n\ge k_1\ge k_2$, and let $\Cs$ be a $[n,k_1+k_2]$ linear code. Then let $\Cs^{\perp}$ be a $(u\mid u+v)$-code construction of $\Cs_1^{\perp}$ and $\Cs_2^{\perp}$, where
\begin{itemize}
\item $\Cs_1{^\perp}$ is an $[n, n-k_1]$ linear code and
  $\Cs_2{^\perp}$ is an $[n, n-k_2]$ linear code, and
\item $\Cs_2^{\perp}\subseteq\Cs_1^{\perp}$.
\end{itemize}
If $\Cs_2$ is a $(n,k,t,m\tau)$ batch code, then $\Cs$ is a $(2n,k_1+k_2,t,m,\tau)$ batch code.
\end{thm}

\begin{proof}
  The first two parameters of $\Cs$ follow from the definition of a
  $(u\mid u+v)$ construction. Let $\Cs$ be constructed as described, and let $\Cs_2$ be an $(n,k,t,m,\tau)$ batch code. Then consider any $t$-tuple of indices $i_1,\dots,i_t\in[2n]$ and let $i_j'=i_j$ if $i_j\in[n]$ and $i_j'=i_j-n$ otherwise. Then we know that $i_1',\dots,i_t'\in[n]$, and so by the batch properties of $\Cs_2$ there exist $t$ disjoint recovery sets for the entries in these indices, the union of which consists of at most $\tau$ entries in each of the $m$ buckets.
  
  If for all $i\in\{n+1,\dots,2n\}$, we place $i$ in the same bucket as $i-n$, then this results in $m$ buckets for $\Cs$. If we consider the recovery sets from above, if $R_{i_j'}$ is the recovery set for $i_j'$, and $i_j\in[n]$, then we claim that $R_{i_j'}'=R_{i_j'}$ is a recovery set for $i_j$ in $\Cs$. This is because the recovery set comes from a vector $v\in\Cs_2^{\perp}$, and by construction, since $\Cs_2^{\perp}\subseteq\Cs_1^{\perp}$, we have that $(v|0),(0|v)\in\Cs^{\perp}$. Likewise, if $i_j\notin[n]$, then by using $(0|v)$, we find that $R_{i_j'}'=\{i+n|i\in R_{i_j'}\}$ is a recovery set for $i_j$.
  
  Since the original $R_{i_j'}$ are all disjoint, so are the $R_{i_j'}'$, and so we have $t$ disjoint recovery sets, the union of which consists of at most $\tau$ elements from each of $m$ buckets, and so $\Cs$ is a $(2n,k_1+k_2,t, m,\tau)$ batch code.
\end{proof}

Because binary Reed-M{\"u}ller codes have parity check matrices that satisfy the above properties, we turn to that family of codes.

\subsection{Locality and availability properties of  $\RM_q(1, \mu)$}\label{s43}
Reed-M{\"u}ller codes for which $\rho = 1$ are known as first order Reed-M{\"u}ller codes. We look at the properties using the polynomial evaluation definition of Reed-M{\"u}ller codes. We begin with a result in the $q$-ary case.


\begin{thm}\label{pointthm}
Let $\F_q^{\mu}=\{P_{i}\mid 1\leq i \leq 2^{\mu}=n\}$ be the set of
evaluation points for $\RM_{q}(1,\mu)$. Then
$(\lambda_{1},\dots,\lambda_{n})\in \RM_{q}(1,\mu)^\perp$  if and only if 
\begin{equation}\label{conditions}
\sum_{i=1}^{n}{\lambda_{i}P_{i}}=0\textrm{ and }\sum_{i=1}^{n}{\lambda_{i}}=0.
\end{equation}
\end{thm}

\begin{proof}
First, if $(\lambda_{1},\dots,\lambda_{n})$ is in the dual code, then by definition,
\begin{equation}
\sum_{i=1}^{n}{\lambda_{i}f(P_{i})}= 0\label{eq6}
\end{equation}
for every polynomial $f\in\F_q[x_1,\dots,x_{\mu}]^1$. In particular, note that for any $1\le k\le\mu$, from $f_k(x_1,\dots,x_\mu)=x_k$, we have
\[\sum_{i=1}^n\lambda_if_k(P_i)=\sum_{i=1}^n\lambda_ip_{i,k}=0,\]
where $p_{i,k}$ is the $k$th entry of point $P_i$. We may gather these equations together for $1\le k\le\mu$ to write the linear combination
\[\sum_{i=1}^n\lambda_iP_i=0.\]
We then consider $f_0=1$, and Equation \eqref{eq6} becomes $\sum_{i=1}^n\lambda_i=0$,
and so we have this direction.

For the other direction, assume that
\[\sum_{i=1}^{n}{\lambda_{i}P_{i}}=0\textrm{ and }\sum_{i=1}^{n}{\lambda_{i}}=0.\]
Then in particular, for any $1\le k\le\mu$, we have $\sum_{i=1}^n\lambda_ip_{i,k}=0$,
and so for any $f\in\F_q[x_1,\dots,x_{\mu}]^1$, by linearity we have that
\[\sum_{i=1}^n\lambda_if(P_i)=\sum_{i=1}^n\lambda_i\left[f_0+\sum_{k=1}^{\mu}f_kp_{i,k}\right]=f_0\sum_{i=1}^n\lambda_i+\sum_{k=1}^{\mu}f_k\sum_{i=1}^n\lambda_ip_{i,k}=0,\]
and thus $(\lambda_1,\dots,\lambda_{\mu})\in\RM(1,\mu)^{\perp}$.
\end{proof}

From Theorem \ref{pointthm} we obtain the following corollaries:

\begin{coro}\label{mwdualg2}
The minimum distance of $\RM_{q}(1,\mu)^{\perp}$ is $4$ if $q = 2$ and $3$ otherwise.
\end{coro}

\begin{proof}
Let $q=2$ and suppose by way of contradiction that the minimum weight is 2. Then there exist two distinct points that sum to zero. This is not possible, and thus the minimum weight must be greater than 2. Note that the only choice of $\lambda_{i}$ is 1, and thus the sum is $0$ if and only if there are an even number of points. Therefore, the weight of the codewords is a multiple of 2, and thus the minimum weight is not 3. The following points are in $\Ps$ (for $\mu\ge2$):
\[P_0=(0,0,0,\dots,0)^T, P_1=(1,0,0,\dots,0)^T, P_2=(0,1,0,\dots,0)^T, \text{ and } P_3=P_1+P_2.
\]
These points satisfy the conditions, and thus the minimum distance for characteristic $2$ is $4$.  

If instead $q\neq 2$, let
$P_1=(0,0,0,0,\dots,0)^T, P_2=(1,0,0,0,\dots,0)^T,
P_3=(-a,0,0,0,\dots,0)^T\in\F_q^\mu$ and the entries of $\lambda$ be
$-(a+1),a$, and $1$ corresponding to the positions of $P_1,P_2,$ and
$P_3$, respectively, and $0$ otherwise. Then, If $a\neq -1,0$,
$\lambda$ satisfies Equations \eqref{conditions}.

Suppose there exists a $\lambda\in\RM_q(1,\mu)^{\perp}$ with weight $2$. Then we have two distinct points $P_i,P_j\in\F_q^{\mu}$ and $\lambda_i,\lambda_j\in\F_q$ such that $\lambda_j=-\lambda_i$ and $\lambda_k=0$ for all $k\ne i,j$. Our two conditions imply:
\[\lambda_iP_i-\lambda_iP_j=0\implies P_i=P_j,\]
A contradiction to the two points being distinct.

\end{proof}

\begin{coro}\label{q2coro}
When $q=2$, every codeword in $\RM(\rho,\mu)$ satisfies Theorem \ref{pointthm} for $\rho\leq\mu-2$.
\end{coro}

\begin{proof}
The dual code of $\RM(1,\mu)$ is $\RM(\mu-2,\mu)$ and $\RM(\rho_{1},\mu)\subset \RM(\rho_{2},\mu)$ if $\rho_{1}<\rho_{2}$. Thus, any codeword in $\RM(\rho,\mu)$ is also in $\RM(\mu-2,\mu)$. Therefore, Theorem \ref{pointthm} implies our claim.
\end{proof}

\begin{thm}\label{athmpg2}
Let $q\ne2$. Then $\RM_{q}(1,\mu)$ has locality $2$ and availability $\delta=\left\lfloor\frac{q^{\mu}-1}{2}\right\rfloor$.
\end{thm}

\begin{proof}
Let $P_a\in\F_q^{\mu}$ be an evaluation point. Then consider any $\alpha\in\F_q$ such that $\alpha\ne0,-1$. We have that $1+\alpha+(-\alpha-1)=0$, and will find corresponding points to use in the reconstruction of $P_a$. For any choice of $P_b\in\F_q^{\mu}$, $P_b\ne P_a$, take
\[P_c=(\alpha+1)^{-1}(P_a+\alpha P_b).\]
Upon rearrangement, we have that $P_a+\alpha P_b+(-\alpha-1)P_c=0$. Furthermore, we find that $P_c\ne P_a,P_b$. If $P_c=P_a$, then our equation becomes $P_a+\alpha P_b+(-\alpha-1)P_a=0$, which simplifies to $\alpha P_b-\alpha P_a=0$, which would contradict $P_b\ne P_a$. Likewise, $P_c=P_b$ would imply $P_a+\alpha P_b+(-\alpha-1)P_b=0$, which becomes $P_a-P_b=0$, another contradiction. A simple counting argument tells us that there are $\left\lfloor\frac{q^{\mu}-1}{2}\right\rfloor$ choices of pairs $P_b,P_c$ for $P_a$, and each of these corresponds to a unique $\lambda\in\RM_q(1,\mu)^{\perp}$ of weight $3$ that can be used to recover $c_a$, the supports of which intersect only in $\{a\}$. Thus, the locality is $2$ and the availability is $\left\lfloor\frac{q^{\mu}-1}{2}\right\rfloor$.
\end{proof}

As proven in \cite{as92} the minimum-weight codewords are generators
of a Reed-M\"uller code $\RM_{p^s}(\mu,\rho)$ where $p$ is a prime
number and $0\leq \rho\leq \mu(p^s-1)$ if and only if either $m=1$ or
$\mu = 1$ or $\rho <p$ or $\rho>(\mu-1)(p^s-1)+p^{s-1}-2.$ Together
with Lemma \ref{RMinfo} it follows that these Reed-M\"uller codes have
all symbol availability. 
Thus, in the following theorems, we only consider the availability of $P_1$ as this implies all symbol availability.
\begin{thm}\label{mueven}
$\RM(1,\mu)$ has availability  $\delta=\frac{2^{\mu}-1}{3}$ when $\mu$ is even. 
\end{thm}
\begin{proof}
  An inductive argument on $\mu$ satisfies this claim. We look at the
  sum of evaluation points to prove our claim. We are looking for
  $\frac{2^{\mu}-1}{3}$ disjoint sets of three points of $\F_q^\mu$
  that sum to $(0,\dots,0)^T\in \F_q^\mu$. 

  It is easy to verify the claim for $\mu=2$ since there is only
  one equation for which this is true
  \[(0,1)^T+(1,0)^T+(1,1)^T=(0,0)^T.\] 
  
  Now assume the claim is true for $\mu=2k$, we show that it is true
  also for $\mu=2k+2$. We have $\frac{2^{2k}-1}{3}$ disjoint sets of
  three points that all sum to $\bar{P}_1=(0,\dots,0)^T\in
  \F_q^{2k}$. Let $P_1=(0,\dots,0)^T\in \F_q^{2k+2}$. For any choice of
  set of points $\{S_1,S_2,S_3\}\subset\F_q^{2k}$ that sum to
  $\bar{P}_1$ in $\F_q^{2k}$ it holds that
\begin{equation}\label{eq17}
\begin{split}
(S_1^T|0,0)^T+(S_2^T|0,0)^T+(S_3^T|0,0)^T&=P_1\\
(S_1^T|1,0)^T+(S_2^T|0,1)^T+(S_3^T|1,1)^T&=P_1\\
(S_1^T|0,1)^T+(S_2^T|1,1)^T+(S_3^T|1,0)^T&=P_1\\
(S_1^T|1,1)^T+(S_2^T|1,0)^T+(S_3^T|0,1)^T&=P_1.
\end{split}
\end{equation}
Additionally it also holds that
\begin{equation}\label{eq17bis}
\begin{split}
(\bar{P}_1^T|1,1)^T+(\bar{P}_1^T|1,0)^T+(\bar{P}_1^T|0,1)^T&=P_1^T.
\end{split}
\end{equation}
The four equations in \eqref{eq17} all use distinct sets of points
because $S_1,S_2,$ and $S_3$ are distinct. Now, there are
$\frac{2^{2k}-1}{3}$ sets of distinct points like $S_1,S_2,$ and
$S_3$. Thus, we have a total of
\[4\cdot\frac{2^{2k}-1}{3}+1 = \frac{2^{2k+2}-1}{3}.\]
Note that the extra one comes from Equation \eqref{eq17bis}. Also note
that our total is an integer as
\[2^{2k+2}\equiv 4^{k+1}\equiv 1^{k+1}\equiv 1\pmod{3}.\]
Because of this, every single coordinate is used, and thus we have maximal availability.
\end{proof}

\begin{thm}\label{muodd}
$\RM(1,\mu)$ has availability $\delta$ at least $\frac{2^{\mu}-4}{4}$ when $\mu$ is odd. 
\end{thm}

\begin{proof}
  We prove this by induction on
  $\mu$. 
  For $\mu=3$, let $S_1=(1,0,0)^T$, $S_2=(0,1,0)^T$ and $S_3=S_1+S_2$,
  then \[S_1+S_2+S_3=\bar{P}_1=(0,0,0)^T.\] No combination of the four remaining
  points of $\F_2^3$ sum to $\bar{P}_1$. Here we have availability $1=\frac{2^3-4}{4}$, and so we have our base case.

 Now assume that, for $\mu=2k+1$, where $k\ge1$, we have that the availability of $\RM(1,\mu)$ is at least $\frac{2^{\mu}-4}{4}$, and there are at least $3$ points that are not used in any recovery set for $P_1$. Let $S_1,S_2$, and $S_3$ be any three points in a recovery set of $P_1\in\F_2^{\mu}$. Then for $\mu+2$, the disjoint sets of three points that
 sum to $\tilde{P}_1=(0,\dots,0)\in\F_2^{\mu+5}$ can be defined by the equations:
\begin{equation*}
\begin{split}
(S_1^T|0,0)^T+(S_2^T|0,0)^T+(S_3^T|0,0)^T&=\tilde{P}_1\\
(S_1^T|1,0)^T+(S_2^T|0,1)^T+(S_3^T|1,1)^T&=\tilde{P}_1\\
(S_1^T|0,1)^T+(S_2^T|1,1)^T+(S_3^T|1,0)^T&=\tilde{P}_1\\
(S_1^T|1,1)^T+(S_2^T|1,0)^T+(S_3^T|0,1)^T&=\tilde{P}_1
\end{split}
\end{equation*}

This results in at least $4\left(\frac{2^{\mu}-4}{4}\right)=2^{\mu}-4$ possible recovery sets. However, we also have points $T_1,T_2$, and $T_3$ that are not used in $\F_2^{\mu}$, and so we may also define the following equations:

\begin{equation*}
\begin{split}
(\bar{P}_1^T|1,0)^T+(T_1^T|0,1)^T+(T_1^T|1,1)^T&=\tilde{P}_1\\
(\bar{P}_1^T|0,1)^T+(T_2^T|1,1)^T+(T_2^T|1,0)^T&=\tilde{P}_1\\
(\bar{P}_1^T|1,1)^T+(T_3^T|1,0)^T+(T_3^T|0,1)^T&=\tilde{P}_1
\end{split}
\end{equation*}

Adding these, we have availability of at least $2^{\mu}-4+3=2^{\mu}-1=\frac{2^{\mu+2}-4}{4}$, and so our property holds for $\mu+2$ as well. We also note that $3(2^{\mu}-1)+3=3\cdot2^{\mu}<4\cdot2^{\mu}=2^{\mu+2}$, and so there are at least $3$ unused points. Thus, by induction, we have that for every $k\ge1$, when $\mu=2k+1$, the availability of $\RM(1,\mu)$ is at least $\frac{2^{\mu}-4}{4}$

Thus we have achieved a lower bound on $\delta$. Note, however, that we have not shown that this is necessarily an optimal construction.
\end{proof}



We now study the batch properties of Reed-M\"uller codes. 

\subsection{Batch Properties of $\RM(1,\mu)$}\label{s432}
\begin{thm}\label{batch14}
The linear code $\RM(1,4)$ is a $(16,5,4,m,\tau)$ batch code for any $m,\tau\in\mathbb{N}$ such that $m\tau=10$.
\end{thm}

\begin{proof}
First, note that the dual code of $\RM(1,4)$ is $\RM(2,4)$, which informs us how to reconstruct elements of the codewords. The generator matrix for $\RM(1,4)$ can be recursively constructed as follows:

\[
G_{1,4}= 
\mat{
G_{1,3} & G_{1,3}\\ 0 & G_{0,3}\\
}
\]

It can be verified that any query of $4$ coordinates of a codeword in $\RM(1,4)$ is possible with the following partition into buckets:
\begin{center}
$\{1\}$, $\{2\}$, $\{3\}$, $\{4\}$, $\{5,6\}$, $\{7,8\}$, $\{9,11\}$, $\{10,12\}$, $\{13,16\}$, $\{14, 15\}$.
\end{center}

In the above case, $m=10$ and $\tau=1$. By Lemma \ref{lemucla2}, this holds for any $m,\tau\in\mathbb{N}$ such that $m\tau=10$.
\end{proof}

We now show how to extend this construction to $\RM(1,\mu)$ for any $\mu\geq 4$.

\begin{thm}\label{batchg4}
Any first order Reed-M{\"u}ller code, $\RM(1,\mu)$, with $\mu\geq 4$, has batch properties $(n,k,4,m,\tau)$ for any $m,\tau\in\mathbb{N}$ such that $m\tau=10$.
\end{thm}

\begin{proof}
We will proceed by induction. First, we have just shown this for the base case where $\mu=4$. Now, assume that for some $\mu-1\ge4$, we have that $\RM(1,\mu-1)$ has batch properties $(n,k,4,m,\tau)$. Recall that the dual code of $\Cs=\RM(1,\mu)$ is $\Cs^{\perp}=\RM(\mu-2,\mu)$. Then as Reed-M{\"u}ller codes follow the $(u \mid u+v)$-construction, $\Cs^{\perp}$ is the $(u\mid u+v)$-code construction of $\Cs_1^{\perp}=\RM(\mu-2,\mu-1)$ and $\Cs_2^{\perp}=\RM(\mu-3,\mu-1)$. Since
$\RM(\mu-3,\mu-1)\subseteq \RM(\mu-2,\mu-1)$, we have $\Cs_2^{\perp}\subseteq\Cs_1^{\perp}$. This means we may apply Theorem $\ref{uvthm}$. Since $\Cs_2=\RM(1,\mu-1)$, we know that $\Cs$ is also a $(n,k,4,m,\tau)$ batch code. By induction, this is now true for any $\mu\ge4$.
\end{proof}

Since the locality of these codes is $r=3$, for $t=4$, we have $m\tau=10=(4-1)\cdot3+1=(t-1)r+1$, and thus we have optimal batch properties. We now extend this even further for most $\RM(\rho,\mu)$.

\subsection{Batch Properties of $\RM(\rho, \mu)$}\label{s44}
We begin with a preliminary result that uses the recursive construction of Reed-M{\"u}ller binary codes.
\begin{lemma}\label{rowspace}
Let $a\in\RM(\rho-1,\mu-2)$. Then
\[(a|a|0|0),(a|0|a|0),(a|0|0|a),(0|a|a|0),(0|a|0|a),(0|0|a|a)\in\RM(\rho,\mu),\]
where $0\in\F_2^{\mu}$.
\end{lemma}

\begin{proof}
Let
\[
G =
\mat{
G_{\rho,\mu-1} & G_{\rho, \mu-1}\\
 0 & G_{\rho-1, \mu-1}\\
}
\]
be the generator of $\RM(\rho,\mu)$. Recursively, we obtain that
\begin{equation}\label{eq26}
G =
\mat{
G_{\rho,\mu-2} & G_{\rho, \mu-2} & G_{\rho,\mu-2} & G_{\rho, \mu-2}\\ 
0 & G_{\rho-1,\mu-2} & 0 & G_{\rho-1,\mu-2}\\
0 & 0 & G_{\rho-1, \mu-2} & G_{\rho-1, \mu-2}\\
0 & 0 & 0 & G_{\rho-2, \mu-2}
}
\end{equation}
From the second and third block rows in matrix (\ref{eq26}), we see that for any $a\in\RM(\rho-1,\mu-2)$, the second row implies $(0|a|0|a)\in\RM(\rho,\mu)$, and the third row implies $(0|0|a|a)\in\RM(\rho,\mu)$. Note that our code is linear, and thus $(0|a|0|a)+(0|0|a|a)=(0|a|a|0)\in\RM(\rho,\mu)$. Finally, note that since $\RM(\rho-1,\mu-2)\subseteq\RM(\rho,\mu-2)$, the first row implies $(a|a|a|a)\in\RM(\rho,\mu)$, and so combining this with the previous vectors, we find that $(a|a|0|0),(a|0|a|0),(a|0|0|a)\in\RM(\rho,\mu)$.
\end{proof}

We now show the batch properties of $\RM(\rho,\mu)$ for $\rho\geq1$ and $\mu\geq 2\rho + 2$.

\begin{thm}\label{batchgen}
  Let $\rho\geq1$ and $\mu\in\{2\rho + 2, 2\rho+3\}$. Then, for
  $\rho'\le\rho$, $\RM(\rho',\mu)$ is a $(n,k,4,m,\tau)$ linear batch
  code for any $m,\tau\in\N$ such that $m\tau=10\cdot2^{2\rho -2}$.
\end{thm}

\begin{proof}
 We focus on the case where
$\mu=2\rho+2$ as the case $\mu=2\rho+3$ proceeds with similar
steps.

If $\RM(\rho, 2\rho+2)$ is a $(n,k,4,m,\tau)$ linear batch code for any $m,\tau\in\N$ such that $m\tau=10\cdot2^{2\rho -2}$, then it follows from Theorem \ref{subspacethm} that for any $\rho'\le\rho$, the code $\RM(\rho',2\rho+2)$ is a $(n,k,4,m,\tau)$ batch code for any $m,\tau\in\N$ such that $m\tau=10\cdot2^{2\rho-2}$. Thus, we need only prove that this property holds for $\RM(\rho, 2\rho+2)$.

We proceed by induction on $\rho$. Note that in Section \ref{s432}, the claim is true for
$\rho=1$, the base cases with $\mu=4,5$. Assume that $\rho\ge1$. We
show that the properties hold for
$\RM(\rho+1,2(\rho+1)+2)=\RM(\rho+1,2\rho+4)$, assuming that
$\RM(\rho,2\rho+2)$ is a $(n,k,4,m,\tau)$ linear batch code for any
$m,\tau\in\N$ such that $m\tau=10\cdot2^{2\rho-2}$. In particular, we
may choose $\tau=1$ and have $m=10\cdot2^{2\rho-2}$ buckets.

We now examine $\RM(\rho+1,2\rho+4)$. The dual code of $\RM(\rho+1,2\rho+4)$ is $\RM(\rho+2,2\rho+4)$. By Lemma \ref{rowspace}, for any $a\in\RM(\rho+1,2\rho+2)=\RM(\rho,2\rho+2)^{\perp}$, we have
\[(a|a|0|0),(a|0|a|0),(a|0|0|a),(0|a|a|0),(0|a|0|a),(0|0|a|a)\in\RM(\rho+2,2\rho+4).\]
This provides a way to produce parity check equations for $\RM(\rho+1,2\rho+4)$ from those for $\RM(\rho,2\rho+1)$, which in turn provides a way to make recovery sets for the former from those for the latter, as each vector corresponds to a recovery set for every index in its support.

Each bucket $B=\{i_1,\dots,i_{\ell}\}$ for $\RM(\rho,2\rho+2)$  can be made into $4$ buckets for $\RM(\rho+1,2\rho+4)$: $B_1=B$, $B_2=B+n=\{i_1+n,\dots,i_{\ell}+n\}$, $B_3=B+2n$, and $B_4=B+3n$. This results in $4\cdot10\cdot2^{2\rho-2}=10\cdot2^{2\rho}=10\cdot2^{2(\rho+1)-2}$ buckets, and we must show that any set of $4$ indices may be recovered by drawing at most $1$ entry from each bucket.

Consider any tuple of $4$ indices $i_1,i_2,i_3,i_4\in[4n]=\bigcup_{s=0}^3([n]+sn)$ and let $s_k=\left\lfloor\frac{i_k-1}{n}\right\rfloor$ for $k\in[4]$. Then let $i_k'=i_k-s_kn$, so that $i_1',i_2',i_3',i_4'\in[n]$. By the induction hypothesis, we have recovery sets $R_1',R_2',R_3',R_4'\subseteq[n]$ for these indices in $\RM(\rho,2\rho+2)$. These recovery sets are sets such that $i_k'\in R_k'$ for all $k\in[4]$ and
\begin{enumerate}
\item $(R_k'\setminus\{i_k'\})\cap(R_j\setminus\{i_j'\})=\emptyset$ for $k,j\in[4]$ with $k\ne j$
\item $\bigcup_{k=1}^4(R_k'\setminus\{i_k'\})$ consists of at most $1$ index in each bucket
\end{enumerate}
Each $R_k'$ is either $\{i_k'\}$ or the support of some vector $a\in\RM(\rho+1,2\rho+2)$. If $|R_k'|=1$, then let $R_k=R_k'+s_kn$. Otherwise, let $R_k=(R_k'+s_kn)\cup(R_k'+s_k'n)$. By Lemma \ref{rowspace}, we know that if $s_k'\ne s_k$, $R_k$ is the support of some vector in $\RM(\rho+2,2\rho+4)$, and so this is a valid recovery set. We must now show that these recovery sets have the desired properties given the correct choice of $s_k'$ values.

We now note that since indices are being recreated from $d=|\{s_1,s_2,s_3,s_4\}|$ different quarters of $[4n]$, we can take at least $d$ of the recovery sets to be singletons. Further, assume that we take as many recovery sets to be singletons as possible. In particular, this means that no recovery set will contain more than one index in each bucket. This is because the only way $R_k$ could contain two indices in a bucket would be if $R_k'$ did. Since $R_k'$ is a proper recovery set, it could only contain two indices in one bucket if one of those indices was $i_k'$. Then that bucket is not used in any other recovery set, and so we could instead take $R_k'=\{i_k'\}$, as per our assumption.

This leaves at most $4-d$ recovery sets which are not singletons and require a subset in a second quarter. Assume without loss of generality that these are $R_1,\dots,R_{d-4}$. Then we may let $s_1',\dots,s_{d-4}'$ be the elements of $\{0,1,2,3\}\setminus\{s_1,s_2,s_3,s_4\}$. Since these are distinct, the only way some $R_k\setminus\{i_k\}$ and $R_j\setminus\{i_j\}$ could have a nonempty intersection would be if condition $1$ was violated. Thus, condition $1$ also holds for the $R_k$. We have already covered the fact that none of the $R_k'+s_k'n$ will not contain more than one index in each bucket, and since these are in separate quarters, the only way $\bigcup_{k=1}^4(R_k\setminus\{i_k\})$ would contain more than one index in a bucket would be if some elements being recovered in the same quarter had $(R_k'\setminus\{i_k'\})\cup(R_j'\setminus\{i_j\})$ consisting of more than 1 index in some bucket. This would violate condition $2$, and so we know that the $R_k$s also satisfy $2$.

Thus, this code is a $(4n,k',4,10\cdot2^{2(\rho+1)-2},1)$ batch code, and by Lemma \ref{lemucla2}, we know that $\RM(\rho+1,2(\rho+1)+2)$ is a $(4n,k',4,m,\tau)$ batch code for any $m,\tau\in\N$ such that $m\tau=10\cdot2^{2(\rho+1)-2}$. This completes the induction step, and so for any $\rho\ge1$, $\RM(\rho,2\rho+2)$ is a $(4n,k',4,m,\tau)$ batch code for any $m,\tau\in\N$ such that $m\tau=10\cdot2^{2\rho-2}$. 
\end{proof}

\section{Conclusion}

This work focuses on batch properties of binary Hamming and
Reed-M\"uller code. 

The high locality of binary Hamming codes implies their availability
to be at most $1$. Binary Hamming codes can be viewed as linear batch codes retrieving queries of at most $2$ indices, the trivial
case. Nonetheless, we prove that for $t=2$, binary Hamming codes
are actually optimal $(2^{s-1}, 2^{s}-1-s, 2, m, \tau)$ batch codes for
$m,\tau\in\N$ such that $m\tau = 2^{s-1}$.

We turn to binary Reed-M\"uller codes for optimal batch codes that
allow larger queries, meaning $t$-tuples with $t$ larger than
$2$. This research direction is motivated by the large availability of
first order Reed-M\"uller codes as showed in the paper. We prove the
optimality of first order Reed-M\"uller codes for $t=4$.
Finally we generalize our study to Reed-M\"uller codes $\RM(\rho,\mu)$
which have
order less than half their length by proving that they have batch
properties $(2^{\mu},k,4,m,\tau)$ such that $m\tau = 10\cdot2^{2\rho
  -2}$ for $\RM(\rho',\mu)$ where $\mu\in\{2\rho+2,2\rho+3\}$ and
$\rho'\le\rho$.




\subsection*{Acknowledgments}

\indent The authors are grateful to Clemson University for hosting the
REU at which this work was completed. The REU was made possible by an
NSF Research Training Group (RTG) grant (DMS \#1547399) promoting
Coding Theory, Cryptography, and Number Theory at Clemson.

\bibliographystyle{plain}
\bibliography{References_new}

\end{document}